\documentclass[12pt]{article}

\usepackage[utf8]{inputenc}
\usepackage[T1]{fontenc}
\usepackage[british]{babel}

\usepackage{amsmath}
\usepackage{amsfonts}
\usepackage{amsthm}

\theoremstyle{plain}
\newtheorem{theorem}{Theorem}
\newtheorem{lemma}{Lemma}

\theoremstyle{remark}
\newtheorem{remark}{Remark}

\usepackage[pdftex]{graphicx}
\usepackage{enumitem}

\begin{document}

\title{Properties of Shannon and R\'{e}nyi entropies  of the
Poisson distribution as the functions of intensity parameter}

\author{Volodymyr Braiman\thanks{Taras Shevchenko National University of Kyiv, 64/13 Volodymyrska St., Kyiv, Ukraine, \texttt{volodymyr.braiman@knu.ua}.} \and Anatoliy Malyarenko\thanks{Mälardalen University, 721 23 Västerås, Sweden, \texttt{anatoliy.malyarenko@mdu.se}.} \and Yuliya Mishura\thanks{Taras Shevchenko National University of Kyiv, 64/13 Volodymyrska St., Kyiv, Ukraine and Mälardalen University, 721 23 Västerås, Sweden, \texttt{yuliyamishura@knu.ua}.} \and Yevheniia Anastasiia Rudyk\thanks{Taras Shevchenko National University of Kyiv, 64/13 Volodymyrska St., Kyiv, Ukraine, \texttt{rudykyao@knu.ua}.}}

\date{\today}

\maketitle

\begin{abstract}
We consider two types of entropy, namely, Shannon and R\'{e}nyi entropies  of the Poisson distribution, and establish their properties  as the functions of intensity parameter. More precisely, we prove that both entropies increase with intensity. While for Shannon entropy the proof is comparatively simple, for R\'{e}nyi entropy, which depends on additional parameter $\alpha>0$, we can characterize it as nontrivial. The proof is based on application of Karamata's inequality to the terms of Poisson distribution.

\textbf{Keywords:} Shannon entropy, R\'{e}nyi entropy, Poisson distribution, Karamata's inequality.
\end{abstract}

\section{Introduction}

The concept and formulas for different types of entropy, which come mainly from information theory, are now widely used, in many applications, in particular,  in detecting DDoS attacks, see, for example, \cite{alada,ddos}, in investigation of  structure of the neural codes,  \cite{paninski},  network traffic \cite{lall}, keyword extraction \cite{singhal}, stock market forecast modelling \cite{karaca} and many others.  Considering these applications, it is natural to assume that, for example,  the distribution of the number of DDoS attacks and of some other related phenomena  is Poisson with some fixed intensity, at least over some fixed time interval. Considering two types of entropy, namely, Shannon and R\'{e}nyi entropy, it is natural to assume that both entropies increase together with intensity parameter. To the best of our knowledge, this statement was not established rigorously. Moreover, if for Shannon entropy the proof is comparatively simple, for R\'{e}nyi entropy, which depends on additional parameter $\alpha>0$, we can characterize it as nontrivial. In order to get the proof, we apply Karamata's inequality to the terms of Poisson distribution. The verification of conditions of this inequality also is nontrivial. It involves rearrangement of the terms of Poisson distribution into a non-increasing sequence and exploring the monotonicity of partial sums of this sequence.

If to  take a historical look, the concept of entropy for a random variable was introduced by Shannon \cite{shan} to characterize the irreducible complexity inherent in a specific form of randomness. Then Shannon entropy was generalized by R\'{e}nyi entropy \cite{renyi} by introducing an additional parameter $\alpha>0$ that allows for a range of entropy measures. The presence of this parameter makes it difficult to accurately calculate the R\'{e}nyi entropy for various distributions and study its behavior. A pleasant exception is the normal distribution. For it, many types of entropies are calculated exactly, and it can be noted that these entropies increase along with the variance (see \cite{entr}). Furthermore, if the decrease of the R\'{e}nyi entropy with respect to the  parameter $\alpha$ is a well-known fact, its convexity for discrete distributions depends on the distribution, see \cite{bur}.

Taking into account the difficulty with exact computation of entropies, many attempts were devoted to the numerical  calculation  and approximation of entropies, see e.g., \cite{archa}--\cite{archer}, \cite{cher}, \cite{han}--\cite{jiao}, \cite{wu}, while the limit  of Shannon entropy of Poisson distribution  was calculated, e.g.,  in \cite{evans}.

As already mentioned, the purpose of this paper is to prove analytically the natural fact that the entropy of a Poisson distribution increases with the intensity of the distribution. We considered Shannon and R\'{e}nyi entropy. The paper is constructed as follows: in Section 2 we consider increase and convexity  of Shannon entropy of Poisson distribution as the function of intensity parameter. In Section 3 increase in intensity   of R\'{e}nyi entropy for any $\alpha>0$ is proved with the help of Karamata's inequality. Some by-product inequalities are obtained in Section 4, while auxiliary statements are postponed to Appendix.

\section{Analytical properties of Shannon entropy of Poisson distribution as the function of intensity parameter $\lambda$}

Consider  a discrete distribution $\{p_i,i\geq1\}.$ Its Shannon  entropy is defined as $$H_{S}(p_i,i\geq1)=-\sum_{i\geq1}p_i \log p_i.$$ In particular, we consider Poisson distribution with parameter $\lambda$:
\[
P\{\xi_{\lambda}=k\}=\frac{\lambda^k e^{-\lambda}}{k!}, k \in \mathbb{N} \cup \{0\}.
\]
Its Shannon entropy $H_S (\lambda)$ equals
\begin{equation}\label{entr1}
H_S (\lambda)=-\sum_{k=0}^{\infty}\frac{\lambda^k e^{-\lambda}}{k!} \log \left(\frac{\lambda^k e^{-\lambda}}{k!}\right)=-\lambda \log \left(\frac{\lambda}{e}\right)+e^{-\lambda} \sum_{k=2}^{\infty} \frac{\lambda^k \log k!}{k!}.
\end{equation}
It is natural to assume that Shannon entropy $H_S (\lambda)$ of Poisson distribution strictly increases with $\lambda \in (0,+\infty).$ We will prove this result, as well as the concavity property of $H_S (\lambda)$, in the next statement.

\begin{theorem} 
Shannon entropy  $H_S (\lambda), \lambda\in (0,+\infty)$ is strictly increasing and concave in $\lambda.$
\end{theorem}
 
\begin{proof} 
Obviously, we can differentiate the series \eqref{entr1} term by term for any $\lambda>0$, and get the equality
 \begin{equation}\begin{split}
 \label{equ1}
 H_S^{'} (\lambda)=&-\log\left(\frac{\lambda}{e}\right)-1-e^{-\lambda}\sum_{k=2}^{\infty} \frac{\lambda^k \log k!}{k!}+e^{-\lambda}\sum_{k=2}^{\infty} \frac{\lambda^{k-1} \log k!}{(k-1)!}\ \\
 =&-\log \lambda+e^{-\lambda}\sum_{k=1}^{\infty} \frac{\lambda^k \log (k+1)!}{k!}-e^{-\lambda}\sum_{k=2}^{\infty} \frac{\lambda^k \log k!}{k!}\ \\
 =&-\log \lambda+e^{-\lambda}\sum_{k=1}^{\infty} \frac{\lambda^k \log (k+1)}{k!}.\
 \end{split}
\end{equation}
 It is clear that both terms in the right-hand side of \eqref{equ1} are non-negative for for $\lambda \in(0,1],$  and the second one is strictly positive, therefore $H_S^{'} (\lambda)>0$ for $\lambda \in(0,1].$ So, it is necessary to prove that $H_S^{'} (\lambda)>0$ for $\lambda >1.$ Let's calculate
  \begin{align*}
 &H_S^{''}(\lambda)=-\frac{1}{\lambda}-e^{-\lambda}\sum_{k=1}^{\infty} \frac{\lambda^k \log (k+1)}{k!}+e^{-\lambda}\sum_{k=1}^{\infty} \frac{\lambda^{k-1} \log (k+1)}{(k-1)!}\ \\
 &=-\frac{1}{\lambda}+e^{-\lambda}\sum_{k=0}^{\infty} \frac{\lambda^k \log (k+2)}{k!}-e^{-\lambda}\sum_{k=1}^{\infty} \frac{\lambda^k \log (k+1)}{k!}\\ \
 &=-\frac{1}{\lambda}+e^{-\lambda} \log 2+e^{-\lambda}\sum_{k=1}^{\infty} \frac{\lambda^k \log (1+\frac{1}{k+1})}{k!}\\ \
 &=-\frac{1}{\lambda}+ e^{-\lambda}\sum_{k=0}^{\infty} \frac{\lambda^k \log (1+\frac{1}{k+1})}{k!}\\ \
 &<-\frac{1}{\lambda}+e^{-\lambda} \sum_{k=0}^{\infty} \frac{\lambda^k}{(k+1)!}<-\frac{1}{\lambda}+e^{-\lambda} \frac{1}{\lambda}\sum_{k=0}^{\infty} \frac{\lambda^{k+1}}{(k+1)!}\\ \
 &<-\frac{1}{\lambda}+e^{-\lambda} \frac{1}{\lambda}e^{\lambda}=0.\
 \end{align*}
So, $H_S^{''}(\lambda)<0$ for all $\lambda>0.$ Therefore, $H_S^{'} (\lambda)$ strictly decreases in $\lambda$ and it is sufficient to prove that $$\lim_{\lambda \to \infty} H_S^{'}(\lambda)\geq 0.$$
However,
\[
\lim_{\lambda \to \infty} H_S^{'}(\lambda)=\lim_{\lambda \to \infty} \log \lambda \left(e^{-\lambda} (\log \lambda)^{-1} \sum_{k=1}^{\infty} \frac{\lambda^k \log(k+1)}{k!} -1\right),
\]
and it is sufficient to establish that
\[
{\liminf}_{\lambda \to \infty} e^{-\lambda} (\log \lambda)^{-1} \sum_{k=1}^{\infty} \frac{\lambda^k \log(k+1)}{k!} \geq 1.
\]
This inequality   is proved in Lemma A.1.
Finally, we get that $H_S^{'} (\lambda)>0$ for all $\lambda\geq 0$ and $H_S^{''} (\lambda)<0$  for all $\lambda\geq 0,$ whence the proof follows.
\end{proof}

\section{Analytical properties of  R\'{e}nyi entropy of Poisson distribution as the function of intensity parameter $\lambda$}

Again, consider a discrete distribution $\{p_i,i\geq1\}.$ Its R\'{e}nyi entropy is defined as
\[
H_{R}^{\alpha}(p_i,i\geq1)=\frac{1}{1-\alpha} \log
\left(\sum_{i\geq1} p_i^{\alpha}\right),~\alpha>0, ~\alpha\neq1,
\]
and
\[
H_{R}^{\alpha}(p_i,i\geq1)\rightarrow H_{S}(p_i,i\geq1),
\]
as  $ \alpha \rightarrow 1$,  where $H_{S}(p_i,i\geq1)$ is a Shannon entropy.
In the case of Poisson distribution
\begin{equation}\label{rpoi}
H_{R}^{\alpha} (\lambda)=\frac{1}{1-\alpha} \log \left(e^{-\alpha\lambda} \sum_{k=0}^{\infty} \frac{\lambda^{k\alpha}}{(k!)^{\alpha}} \right),~\alpha>0, ~\alpha\neq1.
\end{equation}
As for the Shannon entropy, our goal is to investigate the behaviour of   R\'{e}nyi entropy of Poisson distribution as the function of intensity $\lambda.$ To be more precise, we wish to prove that for any fixed $\alpha>0,$ $\alpha\ne1,$ R\'{e}nyi entropy of Poisson distribution increases in $\lambda$. Let us take into account equality \eqref{rpoi} and consider two cases.
\begin{itemize}
 \item Let $\alpha\in(0,1).$ Then $\frac{1}{1-\alpha}>0$ and, since logarithm is strictly increasing, for $H_{R}^{\alpha} (\lambda)$ to increase in $\lambda$ function
\begin{equation}\label{psi}
\psi(\alpha,\lambda)= e^{-\alpha \lambda}\sum_{k=0}^{\infty} \left(\frac{\lambda^k}{k!} \right)^{\alpha}
\end{equation}
should increase in $\lambda\in(0, +\infty)$.
\item Let $\alpha>1.$ Then $\frac{1}{1-\alpha}<0$
and for $H_{R}^{\alpha} (\lambda)$ to increase in $\lambda$ functions
$\log(\psi(\alpha,\lambda))$ and $\psi(\alpha,\lambda)$ should decrease in $\lambda\in(0, +\infty)$.
\end{itemize}

In Theorem \ref{theor2} below we will establish desired character of monotonicity of $\psi(\alpha,\lambda)$ in $\lambda\in(0, +\infty)$ for both cases $0<\alpha<1$ and $\alpha>1.$

To prove this result, we will apply Karamata's inequality (see Lemma A.2 in Appendix). This inequality deals with non-increasing sequences, so it suggests  to rearrange the terms of the Poisson distribution in non-increasing order and study the properties of the resulting sequence.

Let, as before, $p_k(\lambda)=P\{\xi_{\lambda}=k\}=\frac{\lambda^k}{k!}e^{-\lambda}, k \in \mathbb{N} \cup \{0\}.$
Denote by $\{q_k(\lambda): k\ge0\}$ the terms of the sequence $\{p_k(\lambda): k\ge0\}$ which are rearranged in non-increasing order. Note that the sequence
$\{p_k(\lambda): k\ge0\}$ may contain equal terms. In this case the sequence
$\{q_k(\lambda): k\ge0\}$ will contain equal terms as well.

\begin{lemma}\label{majorization} 
For each $n\ge0$ the function $S_n(\lambda)=\sum\limits_{k=0}^n q_k(\lambda)$ strictly decreases on $(0,+\infty).$
\end{lemma}

\begin{proof} Fix $n\ge0.$

Firstly, we will establish that for every $\lambda>0$ the number $S_n(\lambda)$ is equal to the sum of some $n+1$ consecutive terms of the initial sequence
$\{p_k(\lambda): k\ge0\}.$ Note that $\frac{p_{k+1}(\lambda)}{p_k(\lambda)}=\frac{\lambda}{k+1},$ thus $p_k(\lambda)<p_{k+1}(\lambda)$ for $k<\lambda-1$
and $p_k(\lambda)>p_{k+1}(\lambda)$ for $k>\lambda-1.$ Therefore for every
$0\le i_1<i_2<i_3$ we have $p_{i_2}(\lambda)>\min\{p_{i_1}(\lambda),p_{i_3}(\lambda)\}.$ It follows that $q_0(\lambda),\ldots,q_n(\lambda)$ are $n+1$ consecutive terms of the sequence $\{p_k(\lambda): k\ge0\}.$ Thus
$\bigl(q_0(\lambda),\ldots,q_n(\lambda)\bigr)$ is a permutation of terms
$\bigl(p_{\ell}(\lambda),\ldots,p_{\ell+n}(\lambda)\bigr)$ for some
$\ell=\ell(\lambda)\ge0,$ so $S_n(\lambda)=\sum\limits_{k=0}^n q_k(\lambda)=\sum\limits_{k=\ell}^{\ell+n} p_k(\lambda).$

Let us determine $\ell(\lambda).$ Put $s_n(m,\lambda)=\sum\limits_{k=m}^{m+n} p_k(\lambda),$ $m\ge0.$ Clearly
\[
S_n(\lambda)=s_n(\ell,\lambda)=\max\limits_{m\ge0}s_n(m,\lambda).
\]
Note that
\[
s_n(m+1,\lambda)-s_n(m,\lambda)=\sum\limits_{k=m+1}^{m+n+1}p_k(\lambda)-\sum\limits_{k=m}^{m+n}p_k(\lambda)=p_{m+n+1}(\lambda)-p_{m}(\lambda)<0
\]
if and only if 
\[
\frac{p_{m+n+1}(\lambda)}{p_m(\lambda)}=\frac{\lambda^{n+1}m!}{(m+n+1)!}=\frac{\lambda^{n+1}}{(m+1)\ldots(m+n+1)}<1,
\]
i.e., $\lambda<c_m=\bigl((m+1)\ldots(m+n+1)\bigr)^{\frac{1}{n+1}}.$
It follows that $s_n(m+1,\lambda)>s_n(m,\lambda)$ for $\lambda>c_m$ and $s_n(m+1,\lambda)<s_n(m,\lambda)$ for $\lambda<c_m.$ Therefore
$\max\limits_{m\ge0}s_n(m,\lambda)$ is achieved at
\[
\ell(\lambda)=
\begin{cases}
0 & \text{for $0<\lambda\le c_0$},\\ m &\text{for $c_{m-1}\le\lambda\le c_m,$ $m\ge1$}
\end{cases}
\]
(in particular, in case $\lambda=c_k$ we have
$s_n(k+1,\lambda)=s_n(k,\lambda),$ so one can take either $\ell(\lambda)=k$ or $\ell(\lambda)=k+1$).

For $0<\lambda\le c_0$ we have $S_n(\lambda)=\sum\limits_{k=0}^n p_k(\lambda)=\sum\limits_{k=0}^n \frac{\lambda^k}{k!}e^{-\lambda},$
so 
\begin{align*}
S'_n(\lambda)=&\sum\limits_{k=1}^n \frac{\lambda^{k-1}}{(k-1)!}e^{-\lambda}-\sum\limits_{k=0}^n\frac{\lambda^k}{k!}e^{-\lambda}
\\=&\sum\limits_{k=0}^{n-1}\frac{\lambda^{k}}{k!}e^{-\lambda}-\sum\limits_{k=0}^n\frac{\lambda^k}{k!}e^{-\lambda}=-\frac{\lambda^n}{n!}e^{-\lambda}<0
\end{align*}
and the function $S_n(\lambda)$ strictly decreases on $(0,c_0].$

For $c_{m-1}<\lambda\le c_m,$ $m\ge1,$ we have $S_n(\lambda)=\sum\limits_{k=m}^{m+n} p_k(\lambda)=\sum\limits_{k=m}^{m+n} \frac{\lambda^k}{k!}e^{-\lambda},$
so 
\begin{align*}
S'_n(\lambda)=&\sum\limits_{k=m}^{m+n} \frac{\lambda^{k-1}}{(k-1)!}e^{-\lambda}-\sum\limits_{k=m}^{m+n}\frac{\lambda^k}{k!}e^{-\lambda}
\\=&\sum\limits_{k=m-1}^{m+n-1} \frac{\lambda^k}{k!}e^{-\lambda}-\sum\limits_{k=m}^{m+n}\frac{\lambda^k}{k!}e^{-\lambda}=
\frac{\lambda^{m-1}}{(m-1)!}e^{-\lambda}-\frac{\lambda^{m+n}}{(m+n)!}e^{-\lambda}<0,
\end{align*}
which is equivalent to $\lambda^{n+1}>m(m+1)\ldots(m+n),$ or $\lambda>c_{m-1}.$ Hence the function $S_n(\lambda)$ also strictly decreases on $[c_{m-1},c_m]$ for each $m\ge1.$

Clearly $c_m\to+\infty$ as $m\to\infty.$ Thus $S_n(\lambda)$ decreases on \[(0,c_0]\cup \bigcup\limits_{m=1}^\infty[c_{m-1},c_m]=(0,+\infty).\]
\end{proof}

Recall that $\psi(\alpha,\lambda)=\sum\limits_{k=0}^\infty p_k^\alpha(\lambda),$
$\alpha>0,$ $\lambda>0.$

\begin{theorem} \label{theor2}
For every $0<\alpha<1$ the function $\psi(\alpha,\lambda)$ strictly increases as a function of $\lambda$ on $(0,+\infty),$
while for every $\alpha>1$ the function $\psi(\alpha,\lambda)$ strictly decreases as a function of $\lambda$ on $(0,+\infty).$
\end{theorem}

\begin{proof}
Put $r_n(\lambda)=\sum\limits_{k=n+1}^\infty q_k(\lambda),$ $n\ge0.$ Note that
$r_n(\lambda)\to0$ as $n\to\infty$ as a remainder of a convergent series.
It follows from the proof of Lemma \ref{majorization} that for every
$\lambda>0$ there exists $N(\lambda)\in\mathbb{N}$ such that
$\bigl(q_0(\lambda),\ldots,q_n(\lambda)\bigr)$ is a permutation of $\bigl(p_0(\lambda),\ldots,p_n(\lambda)\bigr)$ for each $n\ge N(\lambda).$ Hence
$q_n(\lambda)=p_n(\lambda)$ for every $n\ge N(\lambda).$ Without loss of generality $N(\lambda)\ge2\lambda.$ Then $\frac{p_{k+1}(\lambda)}{p_k(\lambda)}=\frac{\lambda}{k+1}<\frac{1}{2}$ for all $k\ge N(\lambda).$
Thus \[r_n(\lambda)=\sum\limits_{k=n+1}^\infty p_k(\lambda)\le\sum\limits_{k=n+1}^\infty \frac{p_n(\lambda)}{2^{k-n}}=p_n(\lambda)= q_n(\lambda),\; n\ge N(\lambda).\]

Fix any $0<\lambda_1<\lambda_2.$ For every $n\ge N=\max\{N(\lambda_1),N(\lambda_2)\}$ we have
\begin{enumerate}[label=\arabic*)]
\item $q_0(\lambda_i)\ge q_1(\lambda_i)\ge \ldots \ge q_n(\lambda_i)\ge r_n(\lambda_i),\; i=1,2,$ \label{item1}
\item $\sum\limits_{k=0}^i q_k(\lambda_1)>\sum\limits_{k=0}^i q_k(\lambda_2)\;\text{for}\; 0\le i\le n,$ by Lemma \ref{majorization},\label{item2}
\item $\sum\limits_{k=0}^n q_k(\lambda_1)+r_n(\lambda_1)=\sum\limits_{k=0}^n q_k(\lambda_2)+r_n(\lambda_2)=1,$ \label{item3}
\end{enumerate}
i.e. $\bigl(q_0(\lambda_1),\ldots,q_n(\lambda_1),r_n(\lambda_1)\bigr)$ majorizes
$\bigl(q_0(\lambda_2),\ldots,q_n(\lambda_2),r_n(\lambda_2)\bigr)$
(see Lemma A.2 for the definition of majorization).

For $0<\alpha<1$ the function $x^\alpha,$ $x\in [0,1],$ is concave, thus by Karamata's inequality
\begin{equation}\label{ineqpsi1}
\sum\limits_{k=0}^n q_k^\alpha(\lambda_1)+r_n^\alpha(\lambda_1)\le \sum\limits_{k=0}^n q_k^\alpha(\lambda_2)+r_n^\alpha(\lambda_2),\; n\ge N.\end{equation}
Passing to the limit as $n\to\infty,$ we obtain
\begin{equation}\label{ineqpsi2}
\psi(\alpha,\lambda_1)=\sum\limits_{k=0}^\infty q_k^\alpha(\lambda_1)\le \sum\limits_{k=0}^\infty q_k^\alpha(\lambda_2)=\psi(\alpha,\lambda_2),\end{equation}
therefore the function $\psi(\alpha,\lambda)$ increases as a function of $\lambda$ on $(0,+\infty).$

Now we will strengthen inequalities \eqref{ineqpsi1} and \eqref{ineqpsi2} to ensure that $\psi(\alpha,\lambda)$ is strictly monotonic as a function of $\lambda.$

Fix $m$ such that $q_m(\lambda_1)>q_{m+1}(\lambda_1)$ (actually one can take $m=0$ or $m=1$ because the sequence $\{q_k(\lambda_1):k\ge0\}$ contains at most two copies of each term). Put $\widetilde{q}_m(\lambda_1)=q_m(\lambda_1)-\delta,$ $\widetilde{q}_{m+1}(\lambda_1)=q_{m+1}(\lambda_1)+\delta,$ where
$\delta>0$ is such that \[\delta\le\frac{1}{2}\bigl(q_m(\lambda_1)-q_{m+1}(\lambda_1)\bigr)\; \text{and}\; \delta<\sum\limits_{k=0}^m q_k(\lambda_1)-\sum\limits_{k=0}^m q_k(\lambda_2).\]

Note that for every $n\ge \widetilde{N}=\max\{N(\lambda_1),N(\lambda_2),m+1\}$
conditions \ref{item1}--\ref{item3} above remain valid if $q_m(\lambda_1),q_{m+1}(\lambda_1)$ are replaced with $\widetilde{q}_m(\lambda_1),\widetilde{q}_{m+1}(\lambda_1).$ Indeed, only inequalities $\widetilde{q}_m(\lambda_1)\ge\widetilde{q}_{m+1}(\lambda_1)$ in \ref{item1} and $\sum\limits_{k=0}^{m-1} q_k(\lambda_1)+\widetilde{q}_m(\lambda_1)>\sum\limits_{k=0}^m q_k(\lambda_2)$ in \ref{item2} are to be verified anew,
and both of them hold by the choice of $\delta.$ It means that \[\bigl(q_0(\lambda_1),\ldots,q_{m-1}(\lambda_1),\widetilde{q}_m(\lambda_1),\widetilde{q}_{m+1}(\lambda_1),q_{m+2}(\lambda_1),\ldots,q_n(\lambda_1),r_n(\lambda_1)\bigr)\] majorizes
$\bigl(q_0(\lambda_2),\ldots,q_n(\lambda_2),r_n(\lambda_2)\bigr).$
Therefore inequalities \eqref{ineqpsi1} and \eqref{ineqpsi2} remain valid if $q_m(\lambda_1),q_{m+1}(\lambda_1)$ are replaced with $\widetilde{q}_m(\lambda_1),\widetilde{q}_{m+1}(\lambda_1).$
Thus
\[
\sum\limits_{k=0}^\infty q_k^\alpha(\lambda_1)+\widetilde{q}_m^\alpha(\lambda_1)+\widetilde{q}_{m+1}^\alpha(\lambda_1)-q_m^\alpha(\lambda_1)-q_{m+1}^\alpha(\lambda_1)\le \sum\limits_{k=0}^\infty q_k^\alpha(\lambda_2).\]
Since the function $x^\alpha,$ $x\in [0,1],$ is strictly concave on $[0,1],$ we have \[\widetilde{q}_m^\alpha(\lambda_1)+\widetilde{q}_{m+1}^\alpha(\lambda_1)>q_m^\alpha(\lambda_1)+q_{m+1}^\alpha(\lambda_1).\] Hence
\begin{equation}\label{ineqpsi3}
\psi(\alpha,\lambda_1)=\sum\limits_{k=0}^\infty q_k^\alpha(\lambda_1)< \sum\limits_{k=0}^\infty q_k^\alpha(\lambda_2)=\psi(\alpha,\lambda_2).\end{equation}

For $\alpha>1$ the function $x^\alpha,$ $x\in [0,1],$ is convex, thus by Karamata's inequality inequalities \eqref{ineqpsi1}, \eqref{ineqpsi2} and \eqref{ineqpsi3} are reversed.
\end{proof}

\section{By-product inequalities}

As can be seen from Section 3, the proof of increasing (decreasing) properties of function $\psi(\alpha,\lambda)$ did not use   differentiation of  $\psi(\alpha,\lambda)$ and the properties of its derivatives. However, function $\psi(\alpha,\lambda)$  can be differentiated term by term at any point $\alpha>0$ in $\lambda>0$ and as a result, we shall get after some elementary transformations that for any $\alpha, \lambda>0$
\[ 
\psi'_\lambda(\alpha,\lambda)=\alpha e^{-\alpha \lambda}\sum_{k=0}^{\infty} (k-\lambda) \frac{\lambda^{ \alpha k-1}}{(k!)^{\alpha}}.
\]
Obviously, if $\alpha=1$ then for any $\lambda>0$ the derivative is zero because it differs by a strictly positive multiplier $\alpha e^{-\alpha \lambda}$ from the expression 
\[
R(\alpha,\lambda)=\sum_{k=0}^{\infty} (k-\lambda) \frac{\lambda^{ \alpha k-1}}{(k!)^{\alpha}}, 
\]
and  
\[R(1,\lambda)=\sum_{k=1}^{\infty}\frac{\lambda^{  k-1}}{(k-1)!  }-\sum_{k=0}^{\infty}\frac{\lambda^{  k}}{k!  }=0.
\]
However, taking into account Theorem \ref{theor2}, we can establish some nontrivial inequalities.
\begin{lemma}\label{lem2}
For any $\lambda>0$ and $\alpha \in (0,1)$
\begin{equation}\label{equi-1}
\sum_{k=1}^{\infty}\frac{{\lambda}^{\alpha k-1} k^{1-\alpha}}{((k-1)!)^{\alpha}}\ge \sum_{k=0}^{\infty}\frac{{\lambda}^{\alpha k}}{(k!)^{\alpha}},\end{equation}
while for any $\lambda>0$ and $\alpha>1$
\begin{equation}\label{equi-2}
\sum_{k=1}^{\infty}\frac{{\lambda}^{\alpha k-1} k^{1-\alpha}}{((k-1)!)^{\alpha}}\le\sum_{k=0}^{\infty}\frac{{\lambda}^{\alpha k}}{(k!)^{\alpha}}.
\end{equation} 
\end{lemma}

\begin{proof}
Both inequalities follow immediately from the fact that $R(\alpha,\lambda)$ and $\psi'_\lambda(\alpha,\lambda)$ are of the same sign, therefore it follows from Theorem \ref{theor2} that $R(\alpha,\lambda)\ge0$ for all $\lambda>0$ and $\alpha \in (0,1)$ and $R(\alpha,\lambda)\le0$ for all $\lambda>0$ and $\alpha>1$.
\end{proof}

\begin{remark}
 The nontrivial character of inequalities \eqref{equi-1} and \eqref{equi-2} is implied by the fact that we can not compare the respective series term by term, and the relation between $k$th terms depends on whether the condition $\lambda>k$ or the inverse one is satisfied. The similar  situation was with  $\psi(\alpha,\lambda)$, but now, having  Theorem \ref{theor2} in hand, we do not need to analyze the series in more detail in order to compare them.
\end{remark}

\section{Graphical support of results}

In this section we present several plots and surfaces illustrating the behaviour of R\'{e}nyi  entropy as the function of parameters $\alpha$ and $\lambda.$ Obviously, all of them confirm our theoretical results. First, we demonstrate the behaviour of the function $\psi(\alpha,\lambda)$ from \eqref{psi}.

\begin{figure}[htbp]
  \centering
  \includegraphics[width=\columnwidth]{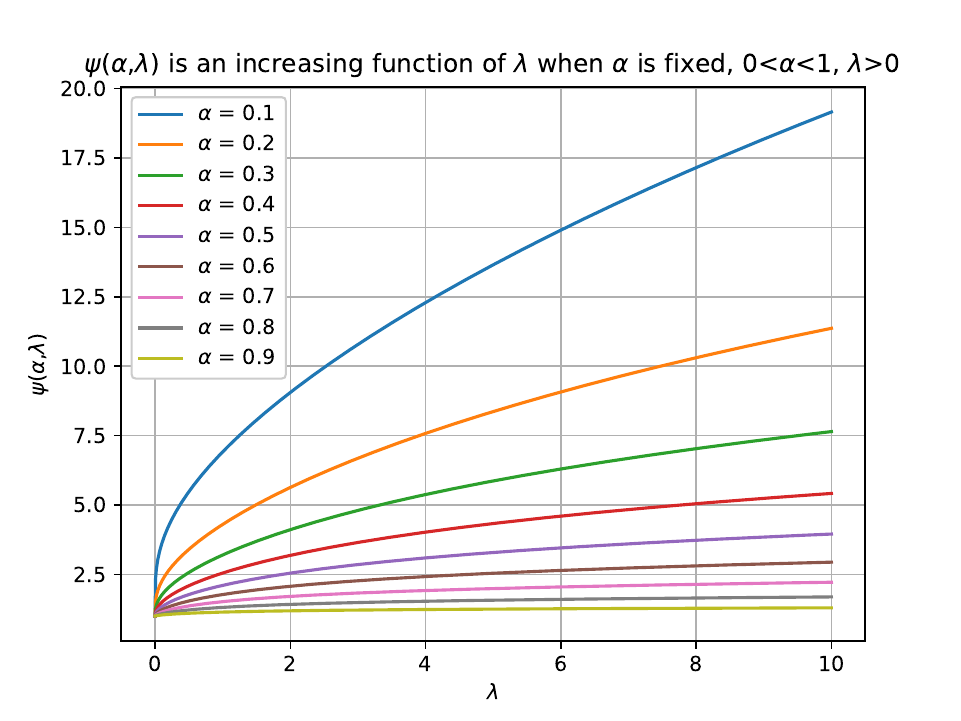}
  \caption{$\psi(\alpha_0,\lambda)$ is an increasing function of $\lambda$ when $\alpha_0$ is fixed, $0<\alpha<1$, $\lambda>0$.}\label{pic1}
\end{figure}

\begin{figure}[htbp]
  \centering
  \includegraphics[width=\columnwidth]{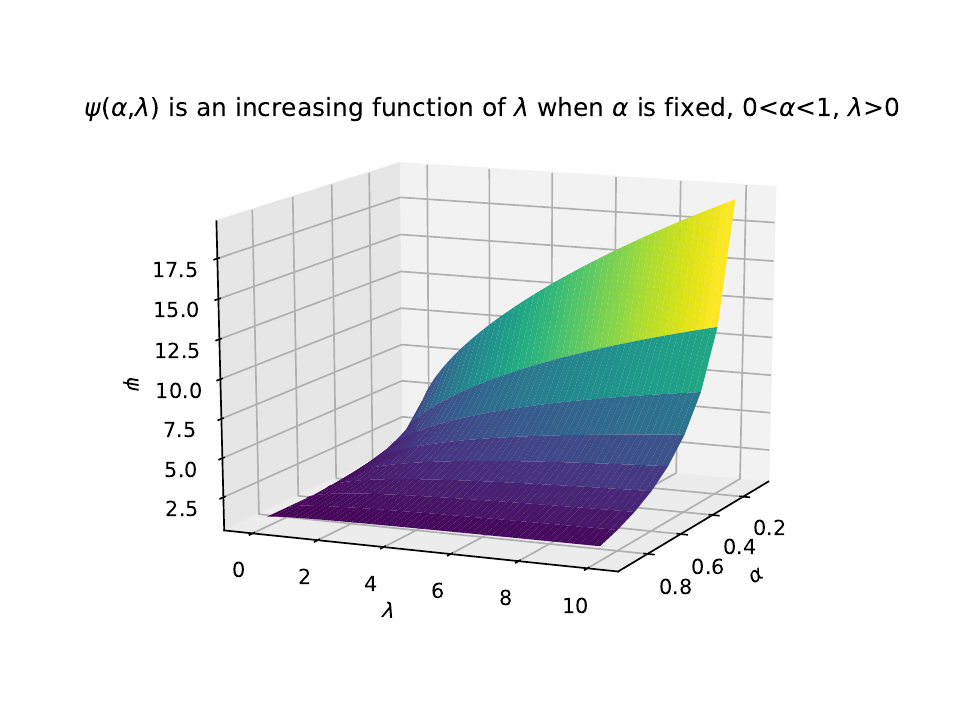}
  \caption{$\psi(\alpha,\lambda)$ is an increasing function of $\lambda$, $0<\alpha<1$, $\lambda>0$.}\label{pic2}
\end{figure}

\begin{figure}[htbp]
  \centering
  \includegraphics[width=\columnwidth]{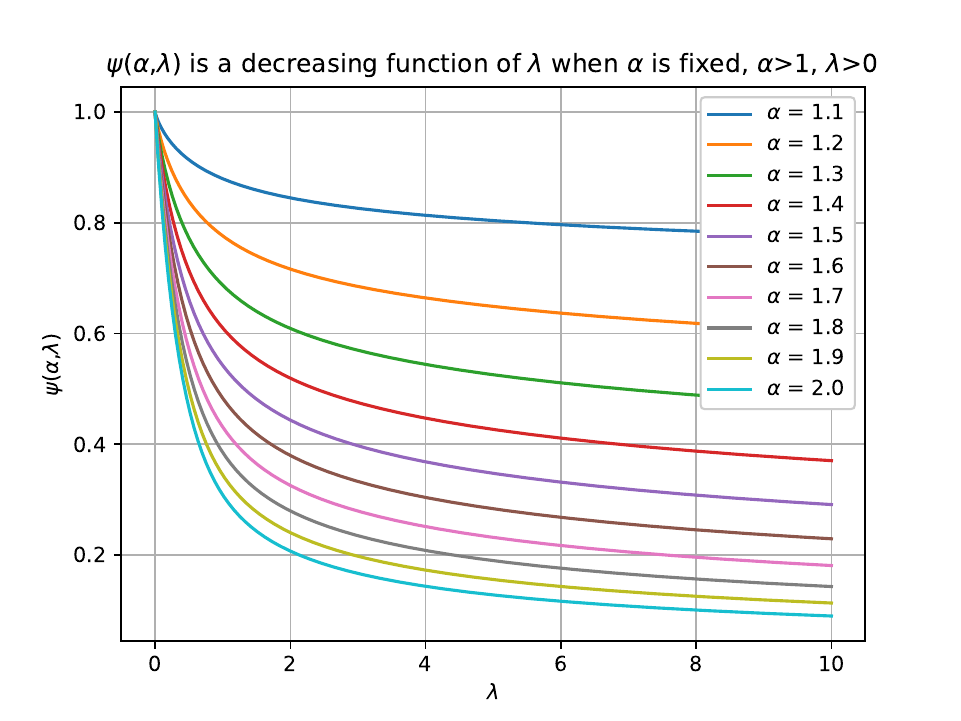}
  \caption{$\psi(\alpha_0,\lambda)$ is an decreasing function of $\lambda$ when $\alpha_0$ is fixed, $\alpha>1$, $\lambda>0$.}\label{pic3}
\end{figure}

\begin{figure}[htbp]
  \centering
  \includegraphics[width=\columnwidth]{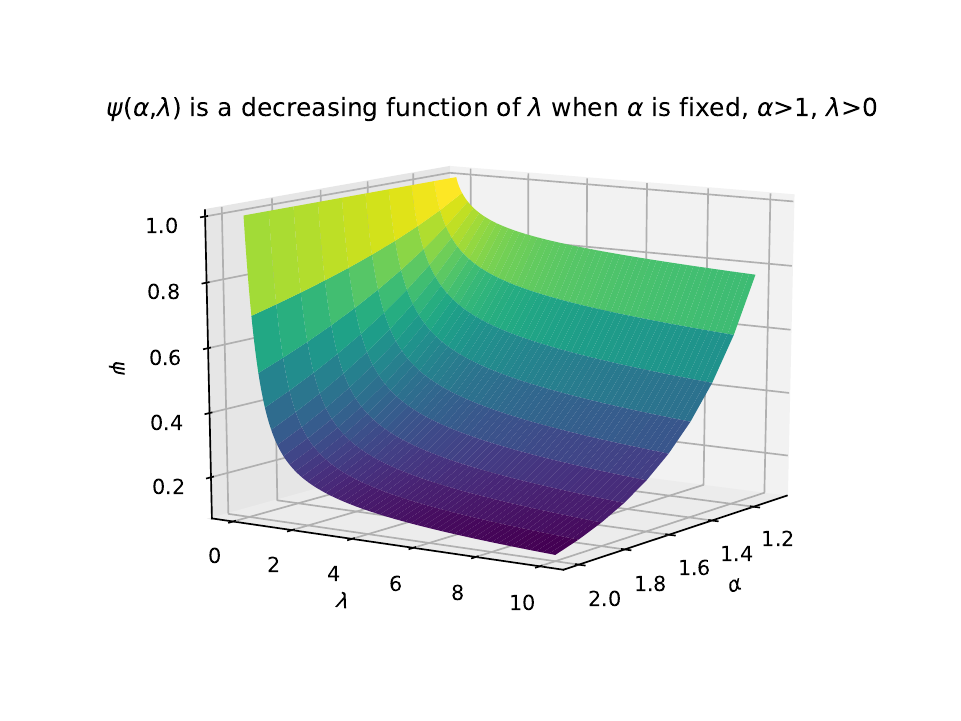}
  \caption{$\psi(\alpha,\lambda)$ is an decreasing function of $\lambda$, $\alpha>1$, $\lambda>0$.}\label{pic4}
\end{figure}

We see from Fig.~\ref{pic1}  that for  fixed $\alpha\in(0, 1) $ (the values of $\alpha$ are chosen from $0.1$ to $0.9$ with interval $0.1$) function $\psi(\alpha,\lambda)$ strictly increases in $\lambda\in(0, +\infty)$. The same is confirmed by the surface on Fig.~\ref{pic2}.

Respectively, we see from Fig.~\ref{pic3} that for  fixed $\alpha>1$ (the values of $\alpha$ are chosen from $1.1$ to $2.0$ with interval $0.1$) function $\psi(\alpha,\lambda)$ strictly decreases  in   $\lambda\in(0, +\infty)$. The same is confirmed by the surface on Fig.~\ref{pic4}.

Now, let us illustrate Lemma \ref{lem2}.

\begin{figure}[htbp]
  \centering
  \includegraphics[width=\columnwidth]{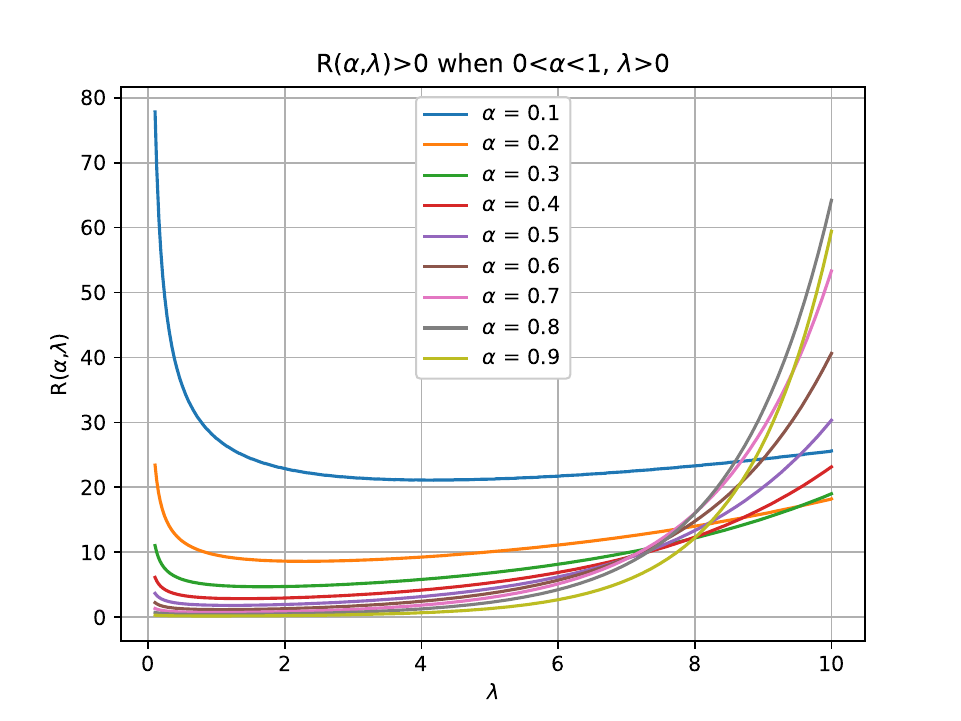}
  \caption{$R(\alpha_0,\lambda)>0$ when $0<\alpha_0<1$ is fixed, $\lambda>0$.}\label{pic5}
\end{figure}

\begin{figure}[htbp]
  \centering
  \includegraphics[width=\columnwidth]{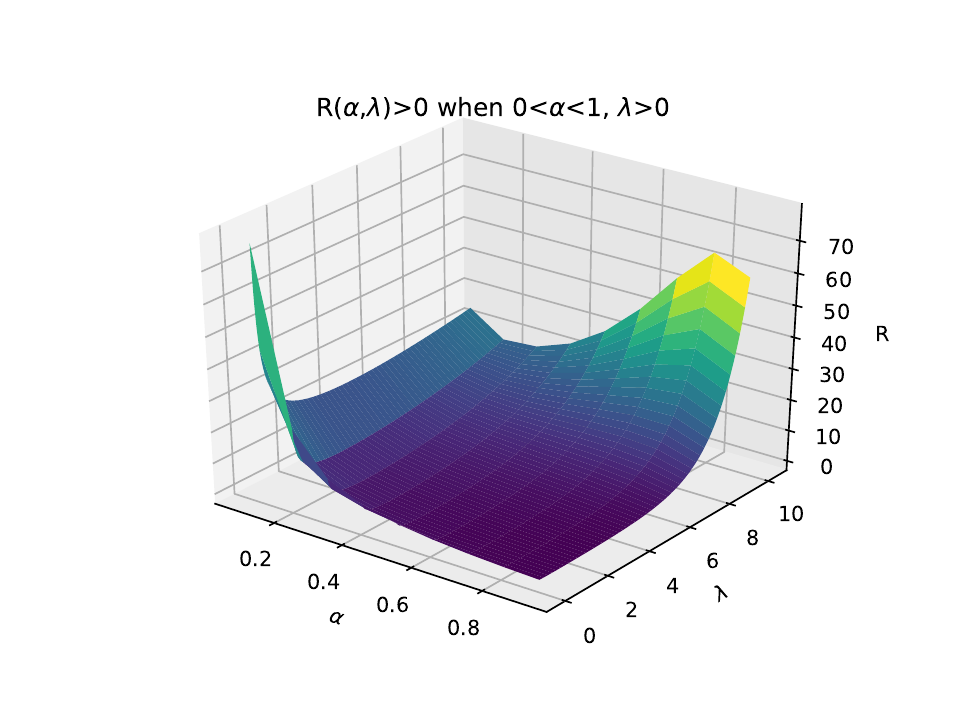}
  \caption{$R(\alpha,\lambda)>0$ when $0<\alpha<1$, $\lambda>0$.}\label{pic6}
\end{figure}

\begin{figure}[htbp]
  \centering
  \includegraphics[width=\columnwidth]{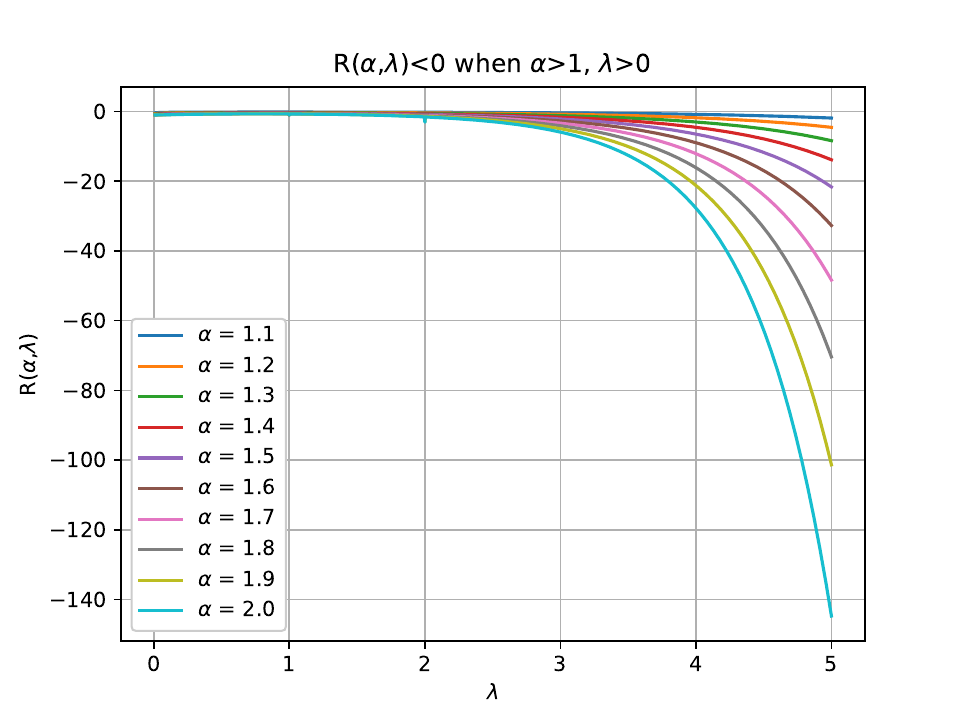}
  \caption{$R(\alpha_0,\lambda)<0$ when $\alpha_0>1$ is fixed, $\lambda>0$.}\label{pic7}
\end{figure}

\begin{figure}[htbp]
  \centering
  \includegraphics[width=\columnwidth]{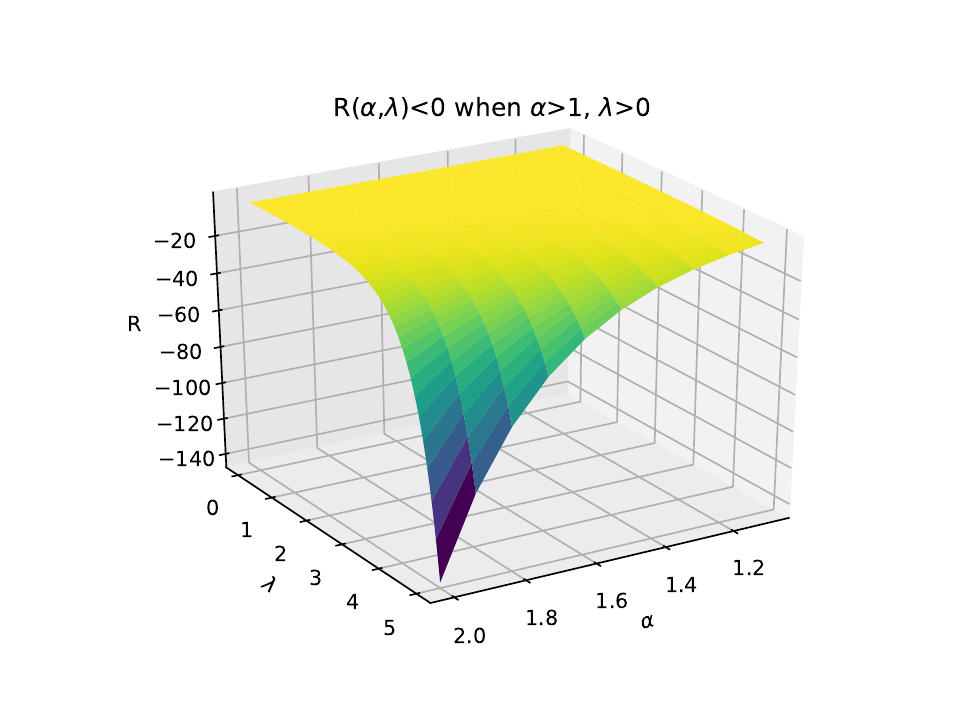}
  \caption{$R(\alpha,\lambda)<0$ when $\alpha>1$, $\lambda>0$.}\label{pic8}
\end{figure}

We see from Fig.~\ref{pic5}  that for $\lambda>0$ and $\alpha\in(0, 1) $ (the values of $\alpha$ are chosen from $0.1$ to $0.9$ with interval $0.1$) $R(\alpha,\lambda)>0$. The same is confirmed by the surface on Fig.~\ref{pic6}.

Respectively, we see from Fig.~\ref{pic7} that for $\lambda>0$ and $\alpha>0$ (the values of $\alpha$ are chosen from $1.1$ to $2.0$ with interval $0.1$) $R(\alpha,\lambda)<0$. The same is confirmed by the surface on Fig.~\ref{pic8}.

\appendix

\section{Appendix}

\begin{lemma}\label{lem1}
 The following relation holds:
\[
\liminf_{\lambda \to \infty} e^{-\lambda} (\log \lambda)^{-1} \sum_{k=1}^{\infty} \frac{\lambda^k \log(k+1)}{k!} \geq 1.
\]
\end{lemma}

\begin{proof} 
It is sufficient to consider $\lambda>4.$
Let us rewrite  the sum under the sign of a limit as follows:
\begin{align*}
S(\lambda):&=\sum_{k=1}^{\infty} \frac{\lambda^k \log(k+1)}{k!}\\ \
 &=\sum_{k=1}^{ \left[\frac{\lambda}{2}\right]} \frac{\lambda^k \log(k+1)}{k!}~
 +\sum_{k\geq[\frac{\lambda}{2}]+1} \frac{\lambda^k \log(k+1)}{k!} \\ \ &=: S_1(\lambda)+S_2(\lambda),
\end{align*}
where $\left[\frac{\lambda}{2}\right]$ is the floor function of $\frac{\lambda}{2}$ (i.e. the greater integer less or equal to $\frac{\lambda}{2}$). Obviously, $$S_1(\lambda) \leq \log \left(\left[\frac{\lambda}{2}\right]+1 \right)\sum_{k=1}^{[\frac{\lambda}{2}]} \frac{\lambda^k}{k!}.$$ Moreover, $\frac{\lambda^k}{k!}$ increases in $k=1,\ldots,[\frac{\lambda}{2}].$ Therefore
\[
S_1(\lambda) \leq \log \left(\left[\frac{\lambda}{2}\right]+1 \right)
\frac{\lambda^{\left[\frac{\lambda}{2}\right]}}{\left[\frac{\lambda}{2}\right]!} \left[\frac{\lambda}{2}\right].
\]
Now, according to two-sided bounds of factorial, for any $n>1$
\[
\sqrt{2\pi n} \left(\frac{n}{e}\right)^n e^{\frac{1}{12n+1}} <n!<\sqrt{2\pi n} \left(\frac{n}{e}\right)^n e^{\frac{1}{12n}}.
\]
So, we can bound $e^{-\lambda} (\log \lambda)^{-1} S_1(\lambda)$ as $\lambda \to \infty,$ as follows:
\begin{align*}0<&e^{-\lambda} (\log \lambda)^{-1} S_1(\lambda) \leq \frac{e^{-\lambda} \lambda^{\left[\frac{\lambda}{2}\right]}}{\left[\frac{\lambda}{2}\right]!} \cdot \frac{\log \left(\left[\frac{\lambda}{2}\right]+1 \right)}{\log \lambda}\left[\frac{\lambda}{2}\right]\\
\leq& \frac{e^{-\lambda}\lambda^{\left[\frac{\lambda}{2}\right]} e^{\left[\frac{\lambda}{2}\right]} \left[\frac{\lambda}{2}\right]}{\sqrt{2\pi} \sqrt{\left[\frac{\lambda}{2}\right]} {\left[\frac{\lambda}{2}\right]}^{\left[\frac{\lambda}{2}\right]} e^{\frac{1}{12\left[\frac{\lambda}{2}\right]+1}}}.\end{align*}
Now shift the range of  the values of $\lambda$ under consideration to $\lambda>42$ and note that   for such $\lambda$ we have the bounds
$\left[\frac{\lambda}{2}\right] \geq \frac{\lambda}{2}-1 \geq \frac{\lambda}{2,1}.$
Therefore for $\lambda>42$
\begin{equation}\label{eq1}
e^{-\lambda} (\log \lambda)^{-1} S_1(\lambda)\leq \frac{e^{\left[\frac{\lambda}{2}\right]-\lambda} (2,1)^{\left[\frac{\lambda}{2}\right]} \left[\frac{\lambda}{2}\right]}{\sqrt{2 \pi} \sqrt{\left[\frac{\lambda}{2}\right]} e^{\frac{1}{12\left[\frac{\lambda}{2}\right]+1}}} <\frac{\left( \frac{2,1}{e}\right)^{\left[\frac{\lambda}{2}\right]} \sqrt{\left[\frac{\lambda}{2}\right]}}{\sqrt{2 \pi}e^{\frac{1}{12\left[\frac{\lambda}{2}\right]+1}}}\xrightarrow[\lambda \to \infty]{}  0.
\end{equation}
Now consider
\[
e^{-\lambda} (\log \lambda)^{-1} S_2(\lambda)\geq e^{-\lambda} \frac{\sum_{k=\left[\frac{\lambda}{2}\right]+1}^{\infty} \frac{\lambda^k}{k!} \log \left(\left[\frac{\lambda}{2}\right]+1 \right)}{\log \lambda}.
\]
Obviously, $\frac{\log \left(\left[\frac{\lambda}{2}\right]+1\right)}{\log \lambda}\to 1$ as $\lambda\to\infty.$ Therefore
\begin{equation}\label{eq2}
\liminf_{\lambda\to\infty} e^{-\lambda} (\log \lambda)^{-1} S_2(\lambda)\geq \liminf_{\lambda\to\infty} \frac{\sum_{k=\left[\frac{\lambda}{2}\right]+1}^{\infty} \frac{\lambda^k}{k!}}{e^{\lambda}}.
\end{equation}
Obviously, $e^{\lambda}=\sum_{k=0}^{\infty} \frac{\lambda^k}{k!},$ while it follows immediately, that similarly to previous calculations in \eqref{eq1}
\begin{equation}\label{eq3}
0 \leq \liminf_{\lambda\to\infty} \frac{\sum_{k=1}^{\left[\frac{\lambda}{2}\right]} \frac{\lambda^k}{k!}}{e^{\lambda}}\leq \limsup_{\lambda\to\infty} \frac{\frac{\lambda^{\left[\frac{\lambda}{2}\right]}}{\left[\frac{\lambda}{2}\right]!}\left[\frac{\lambda}{2}\right]}{e^{\lambda}} \leq \limsup_{\lambda\to\infty} \frac{\left(\frac{2,1}{e} \right)^{\left[\frac{\lambda}{2}\right]} \sqrt{\left[\frac{\lambda}{2}\right]}}{\sqrt{2\pi} e^{\frac{1}{12\left[\frac{\lambda}{2}\right]+1}}}=0.
\end{equation}

Relation (\ref{eq3}) implies that the following limit exists:
\begin{equation}\label{eq4}
\lim_{\lambda\to\infty} \frac{\sum_{k=\left[\frac{\lambda}{2}\right]+1}^{\infty} \frac{\lambda^k}{k!}}{e^{\lambda}}=
\lim_{\lambda\to\infty} \frac{e^{\lambda}-\sum_{k=0}^{\left[\frac{\lambda}{2}\right]} \frac{\lambda^k}{k!}}{e^{\lambda}}=1.
\end{equation}
The proof immediately follows from (\ref{eq1}), (\ref{eq2}) and (\ref{eq4}).

\end{proof}

\begin{lemma}[Karamata's inequality, \cite{kdlm}, \cite{kara}, \cite{moa}]\label{lemma A2} 
Let $f$ be a convex
function on an interval $I\subset\mathbb{R}$ and $a_1,\ldots,a_n,b_1,\ldots,b_n$ be numbers in $I$ such that
$(a_1,\ldots,a_n)$ majorizes $(b_1,\ldots,b_n),$ i.e. the following conditions are fulfilled:\emph{
\begin{enumerate}[label=\arabic*)]
 \item {$a_1\ge a_2\ge\ldots\ge a_n,$ $b_1\ge b_2\ge\ldots\ge b_n;$ }
 \item {$\sum\limits_{k=1}^i a_k\ge \sum\limits_{k=1}^i b_k$ for every $1\le i\le n-1;$}
 \item {$\sum\limits_{k=1}^n a_k=\sum\limits_{k=1}^n b_k.$ }
\end{enumerate}}
Then the inequality $\sum\limits_{k=1}^n f(a_k)\ge \sum\limits_{k=1}^n f(b_k)$ holds.

If $f$ is a concave function on $I$, the inequality is reversed.
\end{lemma}

\paragraph{Acknowledgment.} Yuliya Mishura was supported by   The Swedish Foundation for Strategic Research, grant
Nr. UKR22-0017 and the ToppForsk project nr.
274410 of the Research Council of Norway with title STORM: Stochastics
for Time-Space Risk Models.  We would like to thank Prof. Igor Shevchuk for  the attention to our work and useful advices. 
 
\bibliographystyle{plain}
\bibliography{Poisson}

\end{document}